\documentclass{article}[11pt]
\usepackage{amssymb,verbatim}
\usepackage[letterpaper,hmargin=1in,vmargin=1.25in]{geometry}
\usepackage{graphicx}

\newcommand{\av}{\mathbf{a}}
\newcommand{\bv}{\mathbf{b}}
\newcommand{\cv}{\mathbf{c}}
\newcommand{\kv}{\mathbf{k}}
\newcommand{\uv}{\mathbf{u}}
\newcommand{\vv}{\mathbf{v}}
\newcommand{\xv}{\mathbf{x}}
\newcommand{\yv}{\mathbf{y}}
\newcommand{\zv}{\mathbf{z}}

\newcommand{\am}{\mathbf{A}}
\newcommand{\gm}{\mathbf{G}}
\newcommand{\km}{\mathbf{K}}
\newcommand{\um}{\mathbf{U}}
\newcommand{\vm}{\mathbf{V}}
\newcommand{\xm}{\mathbf{X}}
\newcommand{\zm}{\mathbf{Z}}

\newtheorem{lemma}{Lemma}
\newtheorem{cor}{Corollary}
\newtheorem{theor}{Theorem}
\newtheorem{example}{Example}

\newenvironment{proof}{\noindent {\bf Proof. \ }}{\hfill \vrule height 2pt depth 4pt width 6pt\par\noindent}

\begin{document}

\title{
An Information-Theoretic Analysis \\
of the Security of Communication Systems \\
Employing the Encoding-Encryption Paradigm
}
\author
{Fr\'ed\'erique Oggier and Miodrag J. Mihaljevi\'c
\thanks{
Fr\'ed\'erique Oggier is with Division of Mathematical Sciences,
School of Physical and Mathematical Sciences, Nanyang
Technological University, Singapore. Miodrag Mihaljevi\'c is with
Mathematical Institute, Serbian Academy of Sciences and Arts,
Belgrade, Serbia, and with Research Center for Information
Security (RCIS), Institute of Advanced Industrial Science and
Technology (AIST), Tokyo, Japan.  Email: frederique@ntu.edu.sg,
miodragm@turing.mi.sanu.ac.rs. Part of this work already appeared
at IEEE ICT 2010. } } \maketitle

\begin{abstract}
This paper proposes a generic approach for providing enhanced
security to communication systems which encode their data for
reliability before encrypting it through a stream cipher for
security. We call this counter-intuitive technique the {\em
encoding-encryption} paradigm, and use as motivating example the
standard for mobile telephony GSM. The enhanced security is based
on a dedicated homophonic or wire-tap channel coding that
introduces pure randomness, combined with the randomness of the
noise occurring over the communication channel. Security
evaluation regarding recovery of the secret key employed in the
keystream generator is done through an information theoretical
approach.

We show that with the aid of a dedicated wire-tap encoder, the
amount of uncertainty that the adversary must face about the
secret key given all the information he could gather during
different passive or active attacks he can mount, is a decreasing
function of the sample available for cryptanalysis. This means
that the wire-tap encoder can indeed provide an information
theoretical security level over a period of time, but after a
large enough sample is collected the function tends to zero,
entering a regime in which a computational security analysis is
needed for estimation of the resistance against the secret key
recovery.
\\
{\em Keywords}: error-correction coding, security evaluation,
stream ciphers, randomness, wireless communications, homophonic
coding, wire-tap channel coding.
\end{abstract}

%
%

\section{Introduction}

Most communication systems take into account not only the reliability
but also the security of the data they transmit. This is particularly
true in wireless environment, where the data is inherently more sensible
to security threats. Consequently, the design of such systems need to
include both coding schemes for providing error-correction and
ciphering algorithms for encryption-decryption. It is common practice to first
encrypt the data to ensure its safety, and then to encode it for reliability.
In this paper, we consider the reverse scenario, namely systems which first
encode the data, and then encrypt it, which we call the
{\em encoding-encryption paradigm}.

Though counter-intuitive at first, there are actually many real
life applications where the encoding encryption paradigm is used.
A famous illustrative example is the most widespread standard for
mobile telephony GSM, standing for ``Global System for Mobile
Communications" (see \cite{GSM-coding} and \cite{GSM-encryption},
for the coding, respectively security details). In the GSM
protocol, the data is first encoded using an error-correction code
so as to withstand reception errors, which considerably increases
the size of the message to be transmitted. The encoded data is
then encrypted to provide privacy (secrecy of the communications)
for the users.

It is interesting to mention that block ciphers are not suitable
in the context of the encoding-encryption paradigm, since the
receiver needs to first decrypt the data despite the noise, before
performing the decoding. This leads to use of stream ciphers and
thus when we refer to the security of systems using the
encoding-encryption paradigm, we implicitly mean the security of
the keystream generator and the users' secret key.

From a security perspective, there are of course pros and cons to
the encoding-encryption paradigm. Since it implies encryption of
redundant data (introduced by error-correction), it could be an
origin for mounting attacks against the employed keystream
generator. Undesirability of redundant data from a cryptographic
security point of view has indeed been already pointed out in the
seminal work by Shannon \cite{shannon}, where cryptography as a
scientific topic has been established. On the other hand, the
encoding-encryption paradigm has the advantage to offer protection
in the case of a known plaintext attacking scenario, since an
adversary can only learn a noisy version of the keystream, which
makes the cryptanalyis of the employed keystream generator more
complex.

Security evaluation can be performed under two attacking scenarios, depending
on whether one considers an active or passive adversary.

A {\em passive adversary}'s ability is limited to monitoring (and recording)
communications between the legitimate parties, so as to use the recorded data
as input for mounting a {\em known plaintext attack} against the considered
system.

Stronger attacks come from {\em active adversaries}, which can possibly
include many attacking settings. In this paper, we consider active attacks
motivated by the class of so-called Hopper and Blum (HB) authentication
protocols \cite{hopper-ASIACRYPT2001},\cite{juels-CRYPTO2005},
\cite{katz-EUROCRYPT2006},\cite{gilbert-EUROCRYPT2008},\cite{gilbert-FC2008}.
Following the original work by \cite{hopper-ASIACRYPT2001}, HB authentication
protocols are challenge-response based, where the response could be considered
as the encoded and encrypted version of the challenge, which is deliberately
degraded by random noise. A simple active attack on the improved HB$^+$
authentication protocol \cite{katz-EUROCRYPT2006} was provided in
\cite{gilbert}, where it is assumed that an adversary can manipulate
challenges sent during the authentication exchange, and thus learn whether
such manipulations give an authentication failure. The attack consists of
choosing a constant vector and using it to perturb the challenges by computing
the XOR of the selected vector with each authentication challenge vector, and
that for each of the authentication rounds. To summarize, the active attacker
has the following abilities: (i) he can modify the data in the communication
channel between the legitimate parties; and (ii) he can can learn the effect
of the performed modification at the receiving side.
This is the model that will be adopted in this work.

To evaluate the security of systems using the encoding-encryption
paradigm under threats of both passive and active adversaries as
described above, both computational and information theoretical
analyses are valid. In this paper, we focus on the latter. We
propose a security enhanced approach which employs a dedicated
coding, following the frameworks of homophonic \cite{jendal,
massey-1994, ryabko} and wire-tap channel coding
\cite{wyner,thangaraj-IEEE-IT-2007}. The improved security is a
consequence of combining the pure randomness introduced by the
wire-tap coding and the random noise which is inherent in the
communication channel.

We measure the security increase with respect to the secret key in
terms of its equivocation, that is the amount of uncertainty that
the adversary has on the key, given all the information he can
collect. A preliminary study of the security enhancement has been
provided in \cite{misa&frederique} in the case of a passive
adversary. The enhancement is based on the constructions reported
in \cite{mihaljevic-Computimg2009,mihaljevic-IOSpress2009}, and
also motivated by the fact that in the computational complexity
evaluation scenarios, this approach provides resistance against
the generic time-memory trade-off based attacking approaches
\cite{hellman,mihaljevic-IEEE-CL-2007}, and particular powerful
techniques like the correlation attacks
\cite{fossorier-IEEE-IT-2007}.

\vspace*{0.25cm} \noindent{\em Motivation for the Work}. The aim
of this work is to propose and elaborate a model for the security
evaluation of communication systems which employ the
encoding-encryption paradigm together with a dedicated wire-tap
encoder for security enhancement. In a general security evaluation
scenario, both passive and active attacks should be treated, and
while the enhanced system should be resistant to these, it should
be with a slight/moderate increase of the implementation
complexity and the communications overhead. It may be worth
emphasizing that our target is to increase the security of
existing schemes, such as GSM, which is why we have a small margin
of freedom in designing the security scheme, since we cannot touch
most of the existing components of the system.

\vspace*{0.25cm} \noindent{\em Summary of the Results}. This paper
proposes and analyzes from the information-theoretic point of view
the security of communications systems based on the
encoding-encryption paradigm under passive and active attacks,
when equipped with an additional wire-tap encoder. We show that
with the aid of a dedicated wire-tap encoder, the amount of
uncertainty that the adversary must face about the secret key
given all the information he could gather during different passive
or active attacks he can mount, is a decreasing function of the
sample available for cryptanalysis. This means that the wire-tap
encoder can indeed provide an information theoretical security
level over a period of time, but after a large enough sample is
collected the function tends to zero, entering a regime in which a
computational security analysis is needed for estimation of the
resistance against the secret key recovery.

\vspace*{0.25cm} \noindent{\em Organization of the Paper}. In
Section \ref{sec:model}, we start by describing precisely the
system model together with its security enhanced version and we
dedicate Subsection \ref{subsec:wiretap} to the design of the
wire-tap encoder. The security analysis is done in two parts:
first the passive adversary is studied in Section
\ref{sec:passive}, while the  active one is investigated in
Section \ref{sec:active}. Practical implications of the given
security analysis and some guidelines for design of security
enhanced encoding-encryption based systems are pointed out in
Section \ref{sec:appli}. Concluding remarks including some
directions for future work are given in Section
\ref{sec:conclusion}.

%
%

\section{System Model and Wiretap Coding}
\label{sec:model}

We consider a class of communication systems which, to provide
both reliability and security, employs the encoding and then encryption
paradigm, namely: the message is first encoded, and then encrypted
using a stream ciphering.

\begin{figure}[htb]
\leavevmode
\begin{center}
\includegraphics[width=0.55\hsize]{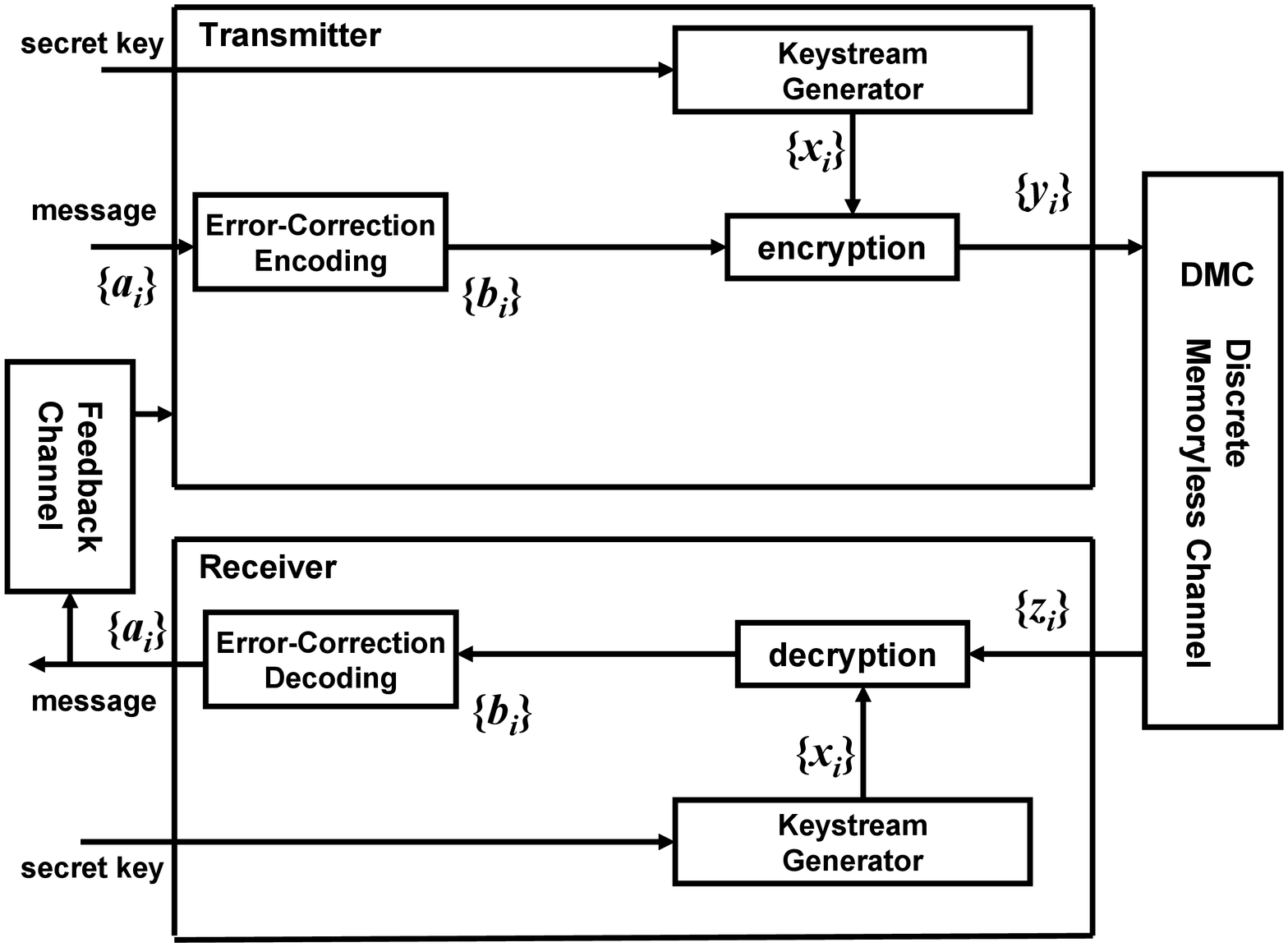}
\caption{Communication system model.}
\label{fig:figure-1}
\end{center}
\end{figure}
The detailed model is shown in Figure \ref{fig:figure-1}.
The transmitter first encodes a binary message/plain text
\[
{\bf a}=[a_i]_{i=1}^m \in \{0,1\}^m
\]
using an error-correcting code $C_{ECC}$
\[
\bv=C_{ECC}(\av)=[b_i]_{i=1}^n\in\{0,1\}^n,
\]
that maps a $m$-dimensional plain text to an $n$-dimensional
encoded message, $n>m$.
The encryption is done using a keystream generator, which takes
as input the secret key $\kv$ of the transmitter, and outputs
\[
{\bf x}=\xv(\kv)=[x_i]_{i=1}^n\in\{0,1\}^n
\]
yielding
\begin{equation}\label{eq:y}
\yv= \yv(\kv)=C_{ECC}(\av)\oplus {\bf x}=[y_i]_{i=1}^n\in\{0,1\}^n
\end{equation}
as the message to be sent over the noisy channel, where
$\oplus$ denotes XOR or modulo 2 addition. We denote the noise
vector by
\[
{\bf v}=[v_i]_{i=1}^{n}\in\{0,1\}^n
\]
where each $v_i$ is the realization of a random variable $V_i$
such that ${\rm Pr}(V_i=1) = p$ and ${\rm Pr}(V_i=0) =1-p$.
Upon reception of the corrupted encrypted binary sequence of ciphertext
\begin{eqnarray*}
\zv &=& \zv(\kv) \\
    &=& \yv + \vv \\
    &=& C_{ECC}(\av)\oplus {\bf x}\oplus \vv=[z_i]_{i=1}^n \in\{0,1\}^n,
\end{eqnarray*}
the receiver who shares the secret key $\kv$ with the transmitter can decrypt
first the message
\[
(C_{ECC}(\av)\oplus {\bf x}\oplus \vv) \oplus \xv
= C_{ECC}(\av)\oplus \vv \in\{0,1\}^n,
\]
and then decode $\av$ despite of the noise thanks to the error-correction code.
We remark that in practice a keystream generator can be considered as a finite
state machine whose initial state is determined by the secret key and some public
data. For simplicity, and because it does not affect our analysis, we can ignore
the existence of the known data, and focus on the secret key. In this setting the
output of the keystream generator is determined uniquely by the secret key, and
it is enough to assume that the transmitter and receiver only share the key.

Note further that the trick of reversing the order of encryption and error-correction
would not have been possible if a block cipher was used for encryption, since decryption
must be done before removing the channel noise.

We finally assume that there is a noiseless feedback link that connects the
receiver to the transmitter, so that the receiver can either acknowledge the
reception of the message, or inform of the decoding failure, so as to get
the missing message sent back.

\subsection{Enhanced model}

Origins for the construction given in this paper are the
approaches for stream ciphers design recently reported in
\cite{mihaljevic-Computimg2009,mihaljevic-IOSpress2009}, though
the focus of this paper is very different, since its goal is
enhancing the security of existing encryption schemes.
This difference has a number of implications regarding the
security issues and implementation complexity of the scheme.

The construction proposed in this paper employs the following main
underlying ideas for enhancing security:

\begin{itemize}

\item Involve pure randomness into the coding\&ciphering scheme
so that the decoding complexity without knowledge of the secret
key employed in the system approaches the complexity of the
exhaustive search for the secret key.

\item  Enhance security of the existing stream cipher via joint
employment of pure randomness and coding theory, and particularly
a dedicated encoding following the homophonic or wire-tap channel
encoding approaches.

\item Allow a suitable trade-off between the security and the
communications rate: Increase the security towards the limit
implied by the secret-key length at the expense of a low-moderate
decrease of the communications rate.

\end{itemize}

Regarding the homophonic and wire-tap channel coding, note the
following. The main goals of homophonic coding are to provide:
(i) multiple substitutions of a given source vector via
randomness so that the coded versions of the source vectors appear
as realizations of a random source; (ii) recoverability of the
source vector based on the given codeword without knowledge of the
randomization. The main goals of wire-tap channel coding are:
(i) amplification of the noise difference between the main and
wire-tap channel via randomness; (ii) a reliable
transmission in the main channel and at the same time to provide a
total confusion of the wire-tapper who observes the communication
in the main channel via a noisy channel (wire-tap channel).
Accordingly, homophonic coding schemes and wire-tap channel ones
have different goals and belong to different coding classes,
the source coding and the error-correction ones, but they employ
the same underlying ideas of using randomness and dedicated
coding for achieving the desired goals.

For enhancing the security we exploit the underlying approaches of
universal homophonic coding \cite{massey-1994} and generic
wire-tap coding when the main channel is error-free (see
\cite{wyner} and \cite{thangaraj-IEEE-IT-2007}, for example).
Accordingly, we may say either
``homophonic coding'' or ``wire-tap channel coding'' to address
the dedicated coding that enhances security. The
main feature of the dedicated coding is that the
encoding is based on randomness and that the
legitimate receiving party who shares a secret key with the
corresponding transmitting one can perform decoding without
knowledge of the randomness employed for the encoding. For
simplicity of the terminology we mainly (but not always) say
``wire-tap channel coding'' to describe the dedicated coding which
provides the enhanced security.

Let $C_H(\cdot)$ denote a wiretap or homophonic code encoder. To
enhance the security of the system considered, it is added at the
transmitter end (see Figure \ref{fig:figure-2}) involving a vector
of pure randomness
\[
{\bf u}=[u_i]_{i=1}^{m-\ell} \in \{0,1\}^{m-l},
\]
that is, each $u_i$ is the realization of a random variable $U_i$ with
distribution ${\rm Pr}(U_i=1) = {\rm Pr}(U_i=0)= 1/2$.
Note that $C_H(\cdot)$ is invertible.
The wiretap encoding is done prior to error-correcting encoding, thus
out of the $m$ bits of data to be sent, $m-l$ are replaced by random data,
letting actually only $l$ bits
\[
{\bf a}=[a_i]_{i=1}^l \in \{0,1\}^l
\]
of plaintext, to get as in (\ref{eq:y})
\begin{equation}\label{eq:ych}
\yv=\yv(\kv)=C_{ECC}(C_H(\av||\uv))\oplus \xv
\end{equation}
as codeword to be sent.
\begin{figure}
\leavevmode
\begin{center}
\includegraphics[width=0.55\hsize]{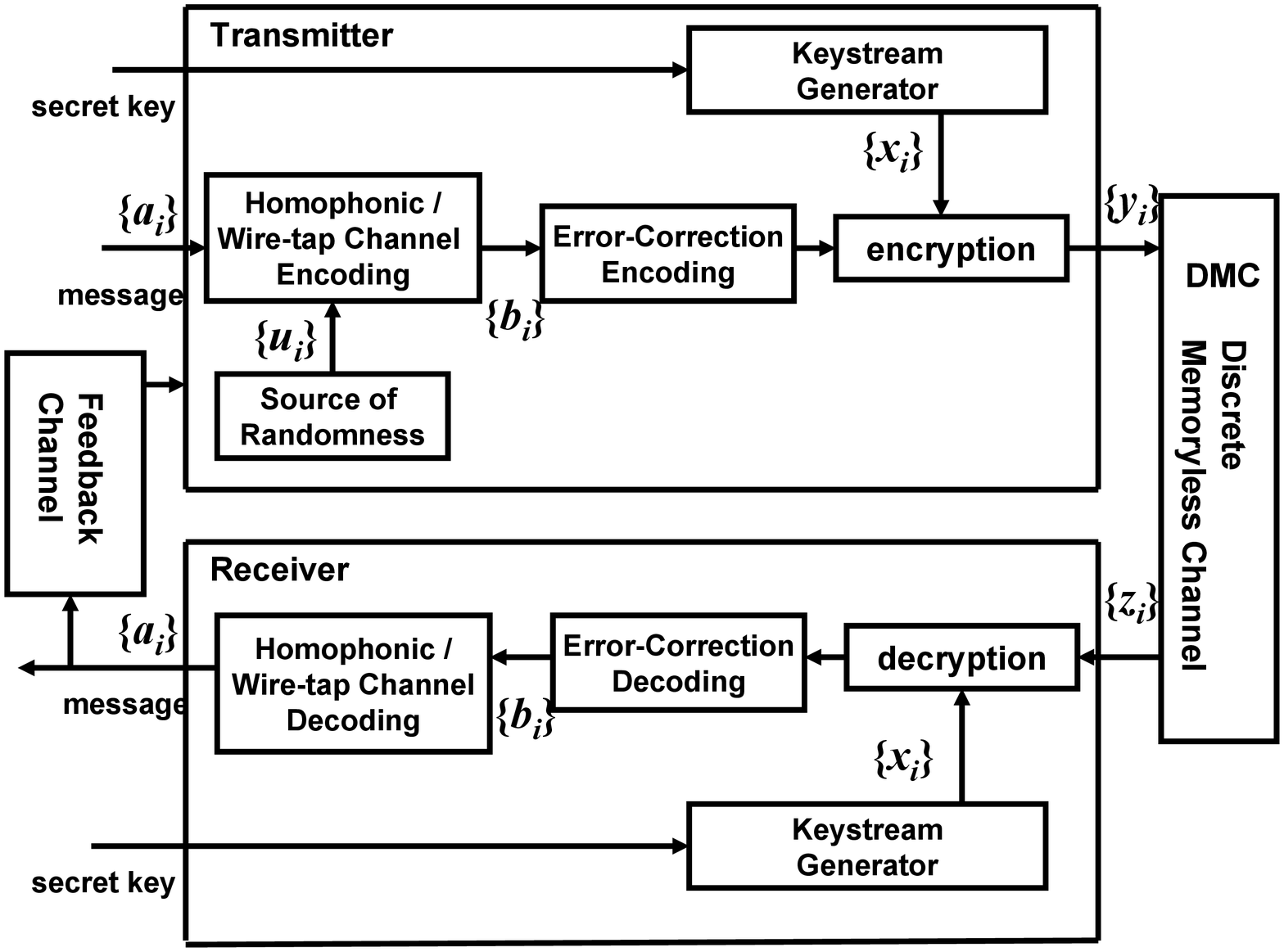}
\caption{Communication system model enhanced with a wire-tap encoder.}
\label{fig:figure-2}
\end{center}
\end{figure}
As before, the receiver obtains
\begin{equation}\label{eq:z}
\zv =\zv(\kv)=\yv \oplus \vv= C_{ECC}(C_H(\av||\uv))\oplus \xv\oplus \vv
\end{equation}
and starts with the decryption
\[
\yv=(C_{ECC}(C_H(\av||\uv))\oplus \xv \oplus \vv)\oplus \xv
=C_{ECC}(C_H(\av||\uv))\oplus \vv.
\]
He then first decodes
\[
C_H(\av||\uv).
\]
If the decoding is successful, he computes $\av$ using $C_H^{-1}$ and
let the transmitter know he could decode. Otherwise he informs the transmitter
than retransmission is required.

Similarly to a linear error-correction code where $C_{ECC}$ can be
represented by multiplying the data vector by the generator matrix of the code,
we can write $C_H$, following the so-called coset encoding proposed by Wyner
\cite{wyner}, as follows:
\begin {equation}\label{eq:coset}
C_H({\bf a}|| {\bf u})=
[{\bf a}|| {\bf u}]
\left[ \begin{array}{l}
{\bf h}_1\\
{\bf h}_2\\
\vdots \\
{\bf h}_l\\
{\bf G}^C
\end{array} \right]
=
 [{\bf a}|| {\bf u}] {\bf G}_H,
\end{equation}
where
\begin{itemize}
\item
${\bf G}^C$ is a $(m-l)\times m$ generator matrix for a $(m,m-l)$ linear
error-correction code $C$ with rows
${\bf g}^C_1,{\bf g}^C_2,\ldots,{\bf g}^C_{m-l}$,
\item
${\bf h}_1,{\bf h}_2,\ldots,{\bf h}_l$ are $l$ linearly independent row
vectors from $\{ 0,1 \}^m \backslash C$,
\item
and ${\bf G}_H$ is a $m \times m$ binary matrix corresponding to
$C_H(\cdot)$.
\end{itemize}
In words, to each $l$-bit message $\av=[a_1,\ldots,a_l]$ is associated
a coset determined by
\[
{\bf a} \mapsto a_1 {\bf h}_1 \oplus a_2 {\bf h}_2 \oplus \ldots
\oplus a_l {\bf h}_l \oplus C.
\]
Though this correspondence is deterministic, a random codeword $\cv$ is
chosen inside the coset by:
\[
{\bf c} = a_1 {\bf h}_1 \oplus a_2 {\bf h}_2 \oplus \ldots \oplus a_l
{\bf h}_l \oplus u_1{\bf g}^C_1 \oplus u_2 {\bf g}^C_2 \oplus \ldots
\oplus u_{m-l} {\bf g}^C_{m-l}
\]
where ${\bf u} = [u_1, u_2,\ldots, u_{m-l}]$ is a uniformly distributed random
$(m-l)$-bit vector.


\subsection{A Dedicated Wiretap Encoder}
\label{subsec:wiretap}

In our scenario, we need to combine wiretap encoding with error-correction
encoding, both being linear operations. Recall that
the encoded vector at the transmitter is
\[
C_{ECC}(C_H({\bf a}||{\bf u})),
\]
where $\av$ is a $l$-dimensional data vector, and $\uv$ is a $m-l$ random
vector. Using generic coset coding as discussed above with a $(m,m-l)$ code,
we now know that
\[
C_H({\bf a}||{\bf u})=[\av||\uv]{\bf G}_H,
\]
where ${\bf G}_H$ is an $m\times m$ matrix, and thus
\begin{eqnarray}
C_{ECC}(C_H({\bf a}||{\bf u}))
&=& C_{ECC}([{\bf a}||{\bf u}] {\bf G}_H) \nonumber \\
&=&[{\bf a}||{\bf u}] {\bf G}_H{\bf G}_{ECC} \nonumber \\
&=& [{\bf a}||{\bf u}] {\bf G} \label{eq:zfinal}
\end{eqnarray}
where  ${\bf G}_{ECC}$ is an $m \times n$ binary generator matrix
corresponding to $C_{ECC}(\cdot)$, and ${\bf G} ={\bf G}_H{\bf G}_{ECC}$
is an $m \times n$ binary matrix summarizing the two successive encodings
at the transmitter.

Since ${\bf G}_H$ multiplies the vector $[\av||\uv]$ where $\av$ is an
$l$-dimension vector and $\uv$ an $(m-l)$ dimension vector, it makes
sense to write the $m\times m$ matrix ${\bf G}_H$ by blocks of size
depending on $l$ and $m-l$:
\begin{equation}\label{RECURRMATR}
{\bf G}_H =
\left[
\begin{array}{cc}
{\bf G}_H^{(1)} & {\bf G}_H^{(2)} \\
{\bf I}_{m-l} & {\bf G}_H^{(4)}
\end{array}
\right]
\end{equation}
where ${\bf G}_H^{(1)}$ is an $l\times (m-l)$ matrix,
${\bf G}_H^{(2)}$ is an $l\times l$ matrix, ${\bf I}_{m-l}$ denotes the
$(m-l)\times (m-l)$ identity matrix, and finally
${\bf G}_H^{(4)}$ is an $(m-l)\times l$ matrix.

Requirements on the matrix ${\bf G}_H$ are:
\begin{enumerate}
\item
{\bf Invertibility.}
The matrix ${\bf G}_H$ should be an invertible matrix, so that the
receiver can decode the wiretap encoding.
\item
{\bf Security.}
The matrix ${\bf G}_H$ should map $[\av||\uv]$ so that in the resulting
vector each bit of data from $\av$ is affected by at least one random
bit from $\uv$, to make sure that each bit of data is protected.
\item
{\bf Sparsity.}
Both the matrices ${\bf G}_H$ and ${\bf G}^{-1}_H$ should be
as sparse as possible, in order to avoid too much computation and
communication overheads.
\end{enumerate}

Since by (\ref{eq:coset}), the $m-l$ last rows of ${\bf G}_H$ form a generator
matrix of a $(m,m-l)$ error correction code $C$ in systematic form, it has rank
$m-l$. The first $l$ rows are then obtained by adding linearly independent vectors
not in $C$, thus completing a basis of $\{0,1\}^m$, resulting automatically in an
invertible matrix. A simple way to do so is to choose
${\bf G}_H^{(1)}={\bf 0}_{l\times(m-l)}$ and ${\bf G}_H^{(2)}={\bf I}_l$, so that
(\ref{RECURRMATR}) becomes
\[
{\bf G}_H =
\left[
\begin{array}{cc}
{\bf 0}_{l\times(m-l)} & {\bf I}_l \\
{\bf I}_{m-l} & {\bf G}_H^{(4)}
\end{array}
\right].
\]

Since
\[
[\av||\uv]
\left[
\begin{array}{cc}
{\bf 0}_{l\times(m-l)} & {\bf I}_l \\
{\bf I}_{m-l} & {\bf G}_H^{(4)}
\end{array}
\right]=[\uv,\av+\uv{\bf G}_H^{(4)}],
\]
and ${\bf G}_H^{(4)}$ has no column with only zeroes (it is a block of an error correction code),
we have that indeed each bit of data from $\av$ is affected by at least one random
bit from $\uv$.

The choice of ${\bf G}_H^{(1)}={\bf 0}_{l\times(m-l)}$ and ${\bf G}_H^{(2)}={\bf I}_l$ makes
the $l$ first rows of ${\bf G}_H$ as sparse as possible.

\begin{example}\rm
Take $m=4$, $l=2$ so that $m-l=2$, and
\[
{\bf G}_H =
\left[
\begin{array}{cc}
{\bf G}_H^{(1)} & {\bf G}_H^{(2)} \\
{\bf I}_2 & {\bf G}_H^{(4)}
\end{array}
\right] =
\left[
\begin{array}{cccc}
0 & 0 & 1 & 0 \\
0 & 0 & 0 & 1 \\
1 & 0 & 1 & 0 \\
0 & 1 & 0 & 1
\end{array}
\right].
\]
Clearly $\gm_H$ is invertible. The error correction code described by rows 3 and 4
is simply the repetition code.
\end{example}

%
%

\section{Security against a Passive Adversary}
\label{sec:passive}

This section analyzes the security of the proposed scheme against
a passive adversary, that is an adversary limited to monitoring
and recording communications. The system we consider already uses
a keystream generator to protect the confidentiality of the data.
Thus though a passive adversary may try to still discover
confidential messages, more dangerous is an attack against the
secret key, which would endanger all the transmissions. Based on
what a passive adversary can do, this means mounting a known
plaintext attack in order to recover the secret key. In the
passive known plaintext attacking scenario, with no enhanced
security, the adversary possesses the pair
\[
(\mbox{plaintext, noisy ciphertext}) = (\av,\zv= C_{ECC}(\av)\oplus\xv\oplus\vv),
\]
from which he calculates
\[
C_{ECC}(\av)\oplus\zv=\xv\oplus\vv.
\]
He can then use $\xv\oplus\vv$ for further processing in an attempt
to recover the key which generated $\xv$.
We will show how the introduction of the wiretap encoding increases the
protection of the key against such attacks.

In what follows, we use as notation that
\begin{itemize}
\item
$u_i$, random bits used in the wiretap encoder,
\item
$x_i$, output bits of the keystream generator,
\item
$v_i$, random components of the additive noise
\end{itemize}
are realizations of certain random variables $U_i$, $X_i$ and $V_i$,
respectively, $i=1,2,...,n$. We can further assume that the plaintext
is generated randomly, and thus see $a_i$ as a realization of a random
variable $A_i$ as well. The corresponding vectors of random variables are
denoted as follows: ${\bf A}^l = [A_i]_{i=1}^{\ell}$,
${\bf U}^{m-l}= [U_i]_{i=1}^{m-\ell}$, ${\bf X}^n= [X_i]_{i=1}^n$, and
${\bf V}^n=[V_i]_{i=1}^n$.

Recall from (\ref{eq:z}) and (\ref{eq:zfinal}) that the received vector at the
receiver is given by
\begin{eqnarray*}
{\bf z} &=&  C_{ECC}(C_H({\bf a}||{\bf u}))\oplus {\bf x}\oplus {\bf v}\\
        &=&  [{\bf a}||{\bf u}] {\bf G}\oplus {\bf x}\oplus {\bf v}
\end{eqnarray*}
where ${\bf G} = [g_{i,j}]_{i=1}^{m} \:_{j=1}^n$ is an $m\times n$ matrix
containing both the wiretap and the error correction encoding.

Let ${\bf z} = [z_i]_{i=1}^n$, so that ${\bf z}$ can be written componentwise as
\[
z_i=((\bigoplus_{k=1}^{\ell} g_{k,i} a_k )\oplus
(\bigoplus_{k=1}^{m-\ell} g_{\ell+k,i} u_k ) \oplus x_i )\oplus v_i,~i=1,2,...,n,
\]
and $z_i$ appears as the realization of a random variable $Z_i$:
\[
Z_i = (( \bigoplus_{k=1}^{\ell} g_{k,i} A_k ) \oplus
(\bigoplus_{k=1}^{m-\ell} g_{\ell+k,i} U_k ) \oplus X_i )\oplus V_i,~ i=1,2,...,n.
\]
We further denote ${\bf Z}^n = [Z_i]_{i=1}^n$, and
\begin{equation}\label{eq:Z}
{\bf Z}^n = C_{ECC}(C_H({\bf A}^l||{\bf U}^{m-l}))\oplus{\bf X}^n
\oplus{\bf V}^n.
\end{equation}
From (\ref{RECURRMATR}), we have
\begin{eqnarray*}
C_H(\am^l||\um^{m-l})
&=&[\am^l,\um^{m-l}]{\bf G}_H\\
&=&[\am^l,\um^{m-l}]
\left[
\begin{array}{cc}
{\bf G}_H^{(1)} & {\bf G}_H^{(2)} \\
{\bf I}_{m-l} & {\bf G}_H^{(4)}
\end{array}
\right]\\
&=&
[\am^l{\bf G}_H^{(1)}, \am^l{\bf G}_H^{(2)}]+
[\um^{m-l},\um^{m-l}{\bf G}_H^{(4)}],
\end{eqnarray*}
and we can rewrite the wiretap encoder as
\[
C_H(\am^l||\um^{m-l})=C_{H,a}(\am^l)\oplus C_{H,u}(\um^{m-l}),
\]
where $C_{H,a}$ and $C_{H,u}$ are the operators for the wiretap encoding
restricted to $\av$, resp. $\uv$:
\[
C_{H,a}(\am^l)=[\am^l{\bf G}_H^{(1)}, \am^l{\bf G}_H^{(2)}],~
C_{H,u}(\um^{m-l})=[\um^{m-l},\um^{m-l}{\bf G}_H^{(4)}].
\]
Since the error correcting encoding is also linear, we finally get
\begin{equation}\label{eq:Z2}
{\bf Z}^n = C_{ECC}(C_{H,a}({\bf A}^l))\oplus C_{ECC}(C_{H,u}({\bf U}^{m-l}))
\oplus{\bf X}^n \oplus {\bf V}^n.
\end{equation}

The lemma below gives a bound on the resistance of the scheme to a known
plain text attack where the adversary knows the pair $({\bf a},{\bf z})$.
\begin{lemma}\label{lem:bound}
The equivocation of the keystream output knowing the plaintext and
the received signal can be lower bounded as follows:
\begin{eqnarray*}
&& H(\xm^n|\am^l,\zm^n)\geq \\
& & \min\{H(\um^{m-l}), H(\xm^n)+H(\vm^n)\} +\\
& & \min\{H(\vm^n),H(\xm^n) \}- \delta(C_{ECC}),
\end{eqnarray*}
where
\begin{eqnarray*}
\delta(C_{ECC}) & = & H(\epsilon)+\epsilon \log(2^{m-l}-1) \\
                &\rightarrow & 0
\end{eqnarray*}
since $\epsilon\rightarrow 0$.
\end{lemma}

\begin{proof}
Employing the entropy chain rule, we have that
\begin{eqnarray*}
& & H(\am^l,\um^{m-l},\xm^n,\vm^n,\zm^n) \\
&=& H(\am^l) + H(\zm^n|\am^l) + H(\um^{m-l}|\am^l,\zm^n)+\\
& & H(\vm^n|\am^l,\um^{m-l},\zm^n) + H(\xm^n|\am^l,\um^{m-l},\vm^n,\zm^n)\\
&=& H(\am^l) + H(\zm^n|\am^l) + H(\um^{m-l}|\am^l,\zm^n)\\
& &  + H(\vm^n|\am^l,\um^{m-l},\zm^n),
\end{eqnarray*}
since $H(\xm^n|\am^l,\um^{m-l},\vm^n,\zm^n)=0$, using that
${\bf X}^n = C_{ECC}(C_H({\bf A}^l||{\bf U}^{m-l}))\oplus{\bf Z}^n \oplus{\bf V}^n$
from (\ref{eq:Z}).

Repeating the entropy chain rule but with another decomposition,
we further get that
\begin{eqnarray*}
& & H(\am^l,\um^{m-l},\xm^n,\vm^n,\zm^n) \\
&=& H(\am^l) + H(\zm^n|\am^l) + H(\xm^n|\am^l,\zm^n)+\\
& & H(\um^{m-l}|\am^l,\xm^n,\zm^n) + H(\vm^n|\am^l,\um^{m-l},\xm^n,\zm^n)\\
&=& H(\am^l) + H(\zm^n|\am^l) + H(\xm^n|\am^l,\zm^n)\\
& &  + H(\um^{m-l}|\am^l,\xm^n,\zm^n),
\end{eqnarray*}
noticing that $H(\vm^n|\am^l,\um^{m-l},\xm^n,\zm^n)=0$ using again
${\bf V}^n = C_{ECC}(C_H({\bf A}^l||{\bf U}^{m-l}))\oplus{\bf Z}^n \oplus{\bf X}^n$
from (\ref{eq:Z}).

By combining the two decompositions, we deduce that
\begin{eqnarray*}
&& H(\xm^n|\am^l,\zm^n)\\
&=& H(\um^{m-l}|\am^l,\zm^n) + H(\vm^n|\am^l,\um^{m-l},\zm^n) \\
&&  - H(\um^{m-l}|\am^l,\xm^n,\zm^n).
\end{eqnarray*}

We now reformulate $H(\um^{m-l}|\am^l,\zm^n)$ and
$H(\vm^n|\am^l,\um^{m-l},\zm^n)$.
First, using this time (\ref{eq:Z2}), we have that
$C_{ECC}(C_{H,u}({\bf U}^{m-l})) = C_{ECC}(C_{H,a}({\bf A}^l))\oplus {\bf Z}^n
\oplus{\bf X}^n \oplus {\bf V}^n$. Since $C_{ECC}$ and $C_H$ are invertible, note
that $H(C_{ECC}(C_{H,u}({\bf U}^{m-l})))=H({\bf U}^{m-l})$, so that
\[
H(\um^{m-l}|\am^l,\zm^n)=H({\bf X}^n \oplus {\bf V}^n ).
\]
On the other hand, conditioning reduces entropy, namely,
\[
H(\um^{m-l}|\am^l,\zm^n)\leq H(\um^{m-l}),
\]
and in order to make explicit the role of the extra randomness brought by the wiretap
encoder, we can write that
\begin{eqnarray*}
H(\um^{m-l}|\am^l,\zm^n) &=& \min \{H(\um^{m-l}), H(\xm^n\oplus\vm^n) \}\\
                        &=& \min \{H(\um^{m-l}), H(\xm^n) + H(\vm^n) \}
\end{eqnarray*}
since ${\bf X}^n$ and ${\bf V}^n$ are mutually independent.

Similarly, again using (\ref{eq:Z2}) to get that ${\bf V}^n= C_{ECC}(C_{H,u}({\bf U}^{m-l}))
\oplus C_{ECC}(C_{H,a}({\bf A}^l))\oplus {\bf Z}^n\oplus{\bf X}^n$ and combining with
\[
H(\vm^n|\am^l,\um^{m-l},\zm^n)\leq H(\vm^n),
\]
we obtain that
\[
H(\vm^n|\am^l,\um^{m-l},\zm^n) = \min \{H(\vm^n), H(\xm^n) \},
\]
which distinguishes the randomness coming from the channel noise
and the keystream entrpy.

We are finally left with bounding $H(\um^{m-l}|\am^l,\xm^n,\zm^n)$.
Recovering ${\bf U}^{m-l}$ when ${\bf A}^l$, ${\bf X}^n$ and
${\bf Z}^n$ are given is the decoding problem of removing the noise
${\bf V}^n$ employing the code $C_{ECC}$ with error probability $P_e$.
This can be bounded using Fano's inequality:
\begin{eqnarray*}
H(\um^{m-l}|\am^l,\xm^n,\zm^n)
&\leq &  H(P_e) + P_e \log(2^{m-l}-1)\\
&\leq &  H(\epsilon)+\epsilon \log(2^{m-l}-1) \rightarrow 0
\end{eqnarray*}
since by design of the system, we may assume $P_e=\epsilon\rightarrow 0$.
This concludes the proof.
\end{proof}

The interpretation of the lemma is a bound on the resistance of
the scheme to a passive known plain text attack. This clearly
depends on two parameters:
\begin{itemize}
\item
the keystream generator: if the output of the keystream generator
has a very high entropy $H(\xm^n)\geq H(\um^{m-l},\vm^n)=H(\um^{m-l})+H(\vm^n)$, then
the lemma tells that
\[
H(\xm^n|\am^l,\zm^n)\geq H(\um^{m-l})+H(\vm^n)-\delta(C_{ECC}).
\]
\item
the pure randomness put in the wiretap encoder: if we do not add
it in the system, the lemma shows that
\[
H(\xm^n|\am^l,\zm^n)\geq H(\vm^n)
\]
that is the information-theoretic security of the keystream
depends on the channel noise.
\end{itemize}
We illustrate this last claim with an example.
\begin{example}\label{ex:kpa}\rm
Consider the case of a known plaintext attack when ${\bf a} = {\bf 0}$.
We then have
\[
z_i= x_i\oplus (\bigoplus_{k=1}^{m-\ell} g_{\ell+k,i} u_k )\oplus v_i,
~i=1,2,...,n.
\]
Without the wiretap encoding, the keystream $x_i$ is corrupted and
so protected as well by the noise on the channel, while with
addition of the wiretap encoder, it is further protected by the
pure randomness added.
\end{example}

The special case where the channel is noisefree is detailed in the
corollary below. This further illustrates the effect of pure randomness
involved in the wire-tap channel coding.
\begin{cor}
In a noisefree channel, we have
\[
H(\xm^n|\am^l,\zm^n) \geq   \min\{H(\um^{m-l}), H(\xm^n)\}.
\]
\end{cor}
\begin{proof}
Since the channel is noisefree, ${\bf V} = 0$ and consequently $H(V)=P_e =0$.
Lemma \ref{lem:bound} can be rewritten as
\begin{eqnarray*}
& & H(\xm^n|\am^l,\zm^n) \\
&\geq & \min\{H(\um^{m-l}), H(\xm^n)\} +\\
&     & \min\{0,H(\xm^n) \}- \delta(C_{ECC})\\
& = & \min\{H(\um^{m-l}), H(\xm^n)\}.
\end{eqnarray*}
\end{proof}

So far, we have discussed the security of a given keystream
generator output, for one instance of transmission. We now move to
a more realistic scenario. Transmission takes place over time
$t=1,2,\ldots$, and the keystream generator uses a secret key (or
just a key) $\km$ based on which it computes its outputs
$\xm^{(t)}=[X_i^{(t)}]_{i=1}^n$ in a deterministic way depending
on $f$ for a time period of length $\tau$:
\[
\xm^{(t)}=\xm^{(t)}(\km)=f^{(t)}(\km),~t=1,\ldots,\tau.
\]
Note that $f^{(t)}(\km)$ is an expansion of the secret key $\km$ via a finite
state machine and can be considered as an encoding of $|\km|$ bits into a long
binary codeword.
Correspondingly, we can rewrite the whole system in terms of realizations
of random variables that depends on time, over the time interval
$t=1,\ldots\tau$:
\begin{itemize}
\item
 ${\bf A}^{(t)} = [A_i^{(t)}]_{i=1}^{\ell}$ for the plain text,
\item
${\bf U}^{(t)} =[U_i^{(t)}]_{i=1}^{m-\ell}$ for the pure randomness used
in the wiretap encoder,
\item
${\bf V}^{(t)}=[V_i^{(t)}]_{i=1}^n$ for the channel noise,
\item
${\bf Z}^{(t)} = [Z_i^{(t)}]_{i=1}^n$ for the received signal.
\end{itemize}
Similarly as above, we have
\[
{\bf Z}^{(t)} = C_{ECC} (C_{H,a}({\bf A}^{(t)}) \oplus
C_{H,u}({\bf U}^{(t)})) \oplus f^{(t)}({\bf K}) \oplus {\bf
V}^{(t)}.
\]
The key ${\bf K}$ is represented as a vector of random variables drawn
independently from a uniform distribution over $\{0,1\}$, so that
$H(\km) = |{\bf K}|$.
We further use the following block notations:
\begin{eqnarray*}
{\bf A}^{\tau l}& = &[{\bf A}^{(1)}|| {\bf A}^{(2)}||\ldots ||{\bf A}^{(\tau)}]\\
{\bf U}^{\tau (m-l)}&=&[{\bf U}^{(1)}||{\bf U}^{(2)}||\ldots||{\bf U}^{(\tau)}]\\
{\bf V}^{\tau n}& = &[{\bf V}^{(1)}|| {\bf V}^{(2)}||\ldots|| {\bf V}^{(\tau)}]\\
{\bf Z}^{\tau n}& = & [{\bf Z}^{(1)}|| {\bf Z}^{(2)}|| ... || {\bf Z}^{(\tau)}].
\end{eqnarray*}

We can now state the main theorem of this section, which describes
the security of the enhanced system against a passive adversary
regarding the secret key recovery.
\begin{theor}\label{theor:pass}
When ${\rm Pr}(V_i^{(j)}=0) \neq {\rm Pr}(V_i^{(j)}=1) \neq 1/2$,
$i=1,2,...,n$, $j=1,2,...,\tau$, there exists a threshold $\tau_{thres}$
such that
\[
H(\km|\am^{\tau l},\zm^{\tau n}) \;\;
\left\{ \begin{tabular}{lll}
$> 0$ & {\rm for} & $\tau < \tau_{thres}$ \\
$\rightarrow 0$   & {\rm for} & $\tau \geq \tau_{thres} \;.$\\
\end{tabular}\right.
\]
\end{theor}

\begin{proof}
When $\tau = 1$, $\xm^{(1)}=\xm^n =f^{(1)}(\km)$ and accordingly
$H(\xm^{(1)}) = H(\km)$, thus Lemma \ref{lem:bound}
directly implies that $H(\xm^n|\am^l,\zm^n)=H(\km|\am^{l},\zm^{n}) > 0$ is
achievable.

When $\tau > 1$ grows, we employ the following analysis.

By using two different decompositions of $H(\am^{\tau
l},\um^{\tau(m-l)},\xm^{\tau n},\vm^{\tau n},\zm^{\tau n})$ via
the entropy chain rule as done in Lemma \ref{lem:bound}, we get
\begin{eqnarray}
&&H(\km|\am^{\tau l},\zm^{\tau n})\nonumber \\
&=& H(\um^{\tau(m-l)}|\am^{\tau l},\zm^{\tau n}) +
   H(\vm^{\tau n}|\am^{\tau l},\um^{\tau(m-l)},\zm^{\tau n})\nonumber \\
&& -H(\um^{\tau(m-l)}|\am^{\tau l},\km,\zm^{\tau n}).
\label{eq:HKAZ}
\end{eqnarray}
Note that knowing $\am^{\tau l}$, ${\bf Z}^{\tau n}$ can be considered as a $\tau
n$-length degraded version of a binary codeword with $\tau
(m-\ell) + |{\bf K}|$ information bits which is corrupted by a
noise vector ${\bf V}^{\tau n}$. Indeed, without knowing the key, decoding $\um^{\tau(m-l)}$ is
not possible, so the adversary also needs to try to decode $\km$. Assuming that the decoding error
probability of this code is $P_e^*$, Fano's inequality implies
that
\begin{displaymath}
H(\um^{\tau(m-l)}|\am^{\tau l},\zm^{\tau n})  <
H(\um^{\tau(m-l)}, \km |\am^{\tau l},\zm^{\tau n})
\end{displaymath}
\begin{displaymath}
\leq H(P_e^*) + P_e^* \log(2^{\tau(m-\ell) + |{\bf K}|}-1) \;.
\end{displaymath}
Combining the decoding ability of $C_{ECC}$ with a minimum
distance decoding yields a decoding error for the aggregated code
of size $2^{\tau(m-\ell) + |{\bf K}|}$ that tends to zero provided
long enough codewords, that is $P_e^* \rightarrow 0$, and accordingly
$H(\um^{\tau(m-l)}|\am^{\tau l},\zm^{\tau n}) \rightarrow 0$ when
$\tau$ is large enough.

In a similar manner and employing
\begin{displaymath}
H(\vm^{\tau n}|\am^{\tau l}\!,\um^{\tau(m-l)}\!,\zm^{\tau n}) <
H(\vm^{\tau n}\!, \km |\am^{\tau l}\!,\um^{\tau(m-l)}\!,\zm^{\tau n}),
\end{displaymath}
the decoding ability of $C_{ECC}$ with a minimum distance
decoding as used above implies that
$H(\vm^{\tau n}|\am^{\tau l},\um^{\tau(m-l)},\zm^{\tau n}) \rightarrow 0$ when
$\tau$ is large enough.

To take care of $H(\um^{\tau(m-l)}|\am^{\tau l},\km,\zm^{\tau n})$,
we again use a decoding argument, since $\zm^{\tau n}$ is known.
However, it is important to note here that $\km$ is known too.
Thus even though we look at a block
\[
{\bf U}^{\tau (m-l)}=[{\bf U}^{(1)}||{\bf U}^{(2)}||\ldots||{\bf
U}^{(\tau)}],
\]
the knowledge of $\km$ makes each block $\um^{(t)}$ independent,
and thus we can decode each of them separately and the probability
of error is $P_e^\tau$. Fano's equality finally yields
\begin{eqnarray*}
H(\um^{\tau(m-l)}|\am^{\tau l},\km,\zm^{\tau n})
\!\!\!&\leq\!\!\! &H(P_e^\tau)+P_e^\tau\log(2^{\tau(m-l)}-1)\\
\!\!\!&\leq\!\!\! &H(\epsilon^\tau)+\epsilon^\tau\log(2^{\tau(m-1)}-1)
\end{eqnarray*}
and
\begin{equation}\label{eq:F}
H(\um^{\tau(m-l)}|\am^{\tau l},\km,\zm^{\tau n}) \rightarrow  0
\end{equation}
since $P_e=\epsilon \rightarrow 0$ by design of $C_{ECC}$.

The above consideration of the cases $\tau=1$ and $\tau >> 1$ also
implies the existence of a threshold $\tau_{thres}$.

\end{proof}

The statement is intuitively clear. The security depends on the
length $|\km|$ of the key noting that this length is fixed in the
system. Accordingly, when the keystream generator is used for a
period $\tau$ that varies, as long as $\tau<\tau_{thresh}$, the
key is protected by the randomness of the noisy channel and of the
wiretap encoder, but that protection cannot last forever if the
adversary collects too much data.

Note that all this analysis is true for ``realistic channels'' where the
noise is not uniformly distributed. The uniformly distributed noise in the
communication channel makes error-correction infeasible, which
explain the assumption in the above theorem.

\noindent Theorem \ref{theor:pass} directly implies the following corollary for noiseless
channels.

\begin{cor} When $\vm^{\tau n} = {\bf
0}$ and the parameter $\tau$ is large enough we have:
\begin{equation}
H(\km |\am^{\tau l},\zm^{\tau n}) = 0 \;. \\
\end{equation}
\end{cor}

%
%

\section{Security against an Active Adversary}
\label{sec:active}

We now consider an active and therefore more powerful adversary.
There are many possible scenarios for an active adversary. We assume in this work
that
\begin{enumerate}
\item
he can modify the data on the communication channel, that is, inject controlled noise,
\item
he can learn the effect of the modified channel at the receiving side, by listening
to the feedback link that tells whether decoding was successful.
\end{enumerate}
Let us be more precise. While the transmitter sends
\[
{\bf y} = C_{ECC}(C_H({\bf a}||{\bf u}))\oplus {\bf x}
\]
in an already security enhances setting (\ref{eq:ych}), the receiver sees
its noisy version
\[
\zv=\yv \oplus \vv.
\]
The active adversary is allowed to inject some extra noise $\vv^*$ over the channel,
so that now, the legitimate receiver sees $\yv\oplus\vv'$, where $\vv'$ contains both the noise $\vv$
coming from the channel and the noise $\vv^*$ controlled by the adversary:
\begin{eqnarray}
{\bf z} & =&  {\bf y} \oplus {\bf v} \oplus {\bf v}^* \nonumber \\
        & =&  C_{ECC}(C_H({\bf a}||{\bf u}))\oplus {\bf x}\oplus
              {\bf v} \oplus {\bf v}^*. \label{eq:zactive}
\end{eqnarray}
As earlier (Section \ref{sec:model}), the receiver first decrypts
its message using its secret key and locally generated keystream
\[
\zv\oplus\xv = C_{ECC}(C_H({\bf a}||{\bf u}))\oplus {\bf v} \oplus{\bf v}^*
\]
and then try to decode $\zv\oplus\xv$:
\begin{eqnarray*}
C_{ECC}^{-1}({\bf z} \oplus {\bf x})
&=& C_{ECC}^{-1}(C_{ECC}(C_H({\bf a}||{\bf u}))\oplus {\bf v} \oplus {\bf v}^* )\\
&=& C_H({\bf a}||{\bf u})
\end{eqnarray*}
under the assumption that the error correcting code can correct the
errors introduced by ${\bf v}$, so as to get
\[
{\bf a} = C_H^{-1}(C_{ECC}^{-1}({\bf z} \oplus {\bf x})).
\]
Because of the extra noise ${\bf v}^*$, the probability of decoding correctly at the
receiver may decrease.
In the meantime, the active attacker can listen to the feedback
channel so that he knows whether the decoding failed or was successful.
His goal is again to find the key. His strategy then consists in adding
different noise vectors ${\bf v}^*$ and to observe the feedback
channel to see whether the chosen noise made the decoding fail, in order to
gather information.

We keep our earlier notation, that is
\begin{itemize}
\item
$u_i$, random bits used in the wiretap encoder,
\item
$x_i$, output bits of the keystream generator,
\item
$v_i'$, random components of the additive noise ${\bf v}' = {\bf v}
\oplus {\bf v}^*$,
\item
$a_i$, bits of the plain text,
\item
$z_i$, bits of the received message
\end{itemize}
are realizations of certain random variables $U_i$, $X_i$, $V_i'$, $V_i$, ${V_i}^*$, $Z_i$, $i=1,2,...,n$
and $A_i$, $i=1,\ldots,l$. The corresponding vectors of random variables are
denoted as follows: ${\bf A}^l = [A_i]_{i=1}^{\ell}$,
${\bf U}^{m-l}= [U_i]_{i=1}^{m-\ell}$, ${\bf X}^n= [X_i]_{i=1}^n$,
${{\bf V}'}^n=[V_i']_{i=1}^n$, ${{\bf V}}^n=[V_i]_{i=1}^n$, ${{\bf V}^*}^n=[{V_i}^*]_{i=1}^n$,
and ${\bf Z}^n = [Z_i]_{i=1}^n$. Similarly to (\ref{eq:Z2}),
\begin{equation}\label{eq:Z2act}
{\bf Z}^n = C_{ECC}(C_{H,a}({\bf A}^l))\oplus C_{ECC}(C_{H,u}({\bf
U}^{m-l})) \oplus{\bf X}^n \oplus {\bf V}^n\oplus{{\bf V}^*}^n.
\end{equation}

Finally, let $f_d$ be a binary flag which indicates whether the
decoding result is indeed ${\bf a}$ or has failed, and accordingly $f_d$ can be considered as a
realization of a binary random variable $F_d$.

The lemma below gives a bound on the resistance of the scheme to an active attack
where the adversary not only controls the noise but also knows ${\bf a},{\bf z}$
and $f_d$.
\begin{lemma}\label{lem:boundactive}
The equivocation of the keystream segment knowing the plaintext,
the received signal, and the decoding tag, can be lower bounded as
follows:
\begin{eqnarray*}
&& H(\xm^n|\am^l,\zm^n, F_d)\geq \\
& & \min \{H(\um^{m-l}), H(\xm^n)+ H(\vm^n|F_d)\}+
\min \{H(\vm^n|F_d), H(\xm^n) \} - \delta(C_{ECC}),
\end{eqnarray*}
where
\[
\delta(C_{ECC})  =  H(P_e)+P_e \log(2^{m-l}-1)\rightarrow 0,
\]
since $P_e\rightarrow 0$.
\end{lemma}

\begin{proof}
As in Lemma \ref{lem:bound}, we start with two different chain rule decompositions
of the same joint entropy. On the one hand,
\begin{eqnarray*}
& & H(\am^l,\um^{m-l},\xm^n,{\vm'}^n,\zm^n, F_d) \\
&=& H(\am^l) + H(\zm^n|\am^l) + H(\um^{m-l}|\am^l,\zm^n)+\\
& & H(F_d|\am^l,\um^{m-l},\zm^n) + H({\vm'}^n|\am^l,\um^{m-l},\zm^n, F_d) + \\
& & H(\xm^n|\am^l,\um^{m-l},{\vm'}^n,\zm^n, F_d)\\
&=& H(\am^l) + H(\zm^n|\am^l) + H(\um^{m-l}|\am^l,\zm^n)\\
& & + H({\vm'}^n|\am^l,\um^{m-l},\zm^n, F_d),
\end{eqnarray*}
since from (\ref{eq:Z2act}) we have that ${\bf X}^n = C_{ECC}(C_H({\bf A}^l||{\bf U}^{m-l}))\oplus{\bf Z}^n
\oplus{{\bf V}'}^n$ implying $H(\xm^n|\am^l,\um^{m-l},{\vm'}^n,\zm^n, F_d)=0$, and
$H(F_d|\am^l,\um^{m-l},\zm^n)=0$, since knowing $\am^l$ and $\zm^n$, decoding can be performed on $\zm$ and the decoded value can be compared to $\am^l$, yielding $F_d$.

On the other hand,
\begin{eqnarray*}
& & H(\am^l,\um^{m-l},\xm^n,{\vm'}^n,\zm^n, F_d) \\
&=& H(\am^l) + H(\zm^n|\am^l) + H(\xm^n|\am^l,\zm^n)+\\
& & H(F_d|\am^l,\xm^n,\zm^n) + H(\um^{m-l}|\am^l,\xm^n,\zm^n, F_d) + \\
& & H({\vm'}^n|\am^l,\um^{m-l},\xm^n,\zm^n, F_d)\\
&=& H(\am^l) + H(\zm^n|\am^l) + H(\xm^n|\am^l,\zm^n)\\
& & + H(\um^{m-l}|\am^l,\xm^n,\zm^n, F_d),
\end{eqnarray*}
noticing that $H({\vm'}^n|\am^l,\um^{m-l},\xm^n,\zm^n, F_d)=0,$ again
using from (\ref{eq:Z2act}) that ${{\bf V}'}^n = C_{ECC}(C_H({\bf A}^l||{\bf U}^{m-l}))\oplus{\bf Z}^n
\oplus{\bf X}^n$, and that $H(F_d|\am^l,\xm^{m-l},\zm^n)=0$ for the same reason as above.

By combining the two decompositions, we deduce that
\begin{equation}\label{eq:hxaz}
H(\xm^n|\am^l,\zm^n)
= H(\um^{m-l}|\am^l,\zm^n,F_d) + H({\vm'}^n|\am^l,\um^{m-l},\zm^n, F_d)
 - H(\um^{m-l}|\am^l,\xm^n,\zm^n, F_d),
\end{equation}
where
\[
H(\um^{m-l}|\am^l,\zm^n, F_d) =  \min \{H(\um^{m-l}), H(\xm^n\oplus{\vm'}^n|F_d) \}
\]
since $C_{ECC}(C_H({\bf U}^{m-l}))=C_{ECC}(C_H({\bf A}^l))\oplus{\bf Z}^n
\oplus{\bf X}^n\oplus {{\bf V}'}^n$ implies that
\[
H(\um^{m-l}|\am^l,\zm^n, F_d) =H(\xm^n\oplus{\vm'}^n|F_d)
\]
and conditioning reduces entropy, namely,
\[
H(\um^{m-l}|\am^l,\zm^n, F_d)\leq H(\um^{m-1}).
\]

Similarly, again using (\ref{eq:Z2act}) and that
\[
H({\vm'}^n|\am^l,\um^{m-l},\zm^n, F_d)\leq H({\vm'}^n|F_d),
\]
we obtain that
\[
H({\vm'}^n|\am^l,\um^{m-l},\zm^n, F_d) =\min \{H({\vm'}^n|F_d), H(\xm^n) \}.
\]

To summarize, Equation (\ref{eq:hxaz}) is now given by
\begin{eqnarray*}
&&H(\xm^n|\am^l,\zm^n)\\
&=&\min \{H(\um^{m-l}), H(\xm^n\oplus{\vm'}^n|F_d) \}+
    \min \{H({\vm'}^n|F_d), H(\xm^n) \}\\
&&-H(\um^{m-l}|\am^l,\xm^n,\zm^n,F_d)
\end{eqnarray*}
where ${\vm'}^n=\vm^n+{\vm^*}^n$ and ${\vm^*}^n$ is known to the adversary, so
that we in fact have
\begin{eqnarray*}
&&H(\xm^n|\am^l,\zm^n)\\
&=& \min \{H(\um^{m-l}), H(\xm^n)+ H(\vm^n|F_d)\}+
\min \{H(\vm^n|F_d), H(\xm^n) \}\\
&&-H(\um^{m-l}|\am^l,\xm^n,\zm^n,F_d)
\end{eqnarray*}
using further that ${\bf X}^n$ and ${\bf V}^n$ are mutually independent.

We are finally left with bounding $H(\um^{m-l}|\am^l,\xm^n,\zm^n,
F_d)$. For the adversary, recovering ${\bf U}^{m-l}$ when ${\bf A}^l$, ${\bf X}^n$ and
${\bf Z}^n$ are given is the decoding problem of removing the noise
${\bf V}^n$ (he knows ${{\bf V}^*}^n$) employing the code $C_{ECC}$ with error probability
$P_e$. This can be bounded using Fano's inequality:
\begin{eqnarray*}
H(\um^{m-l}|\am^l,\xm^n,\zm^n, F_d)
&=& H(\um^{m-l}|\am^l,\xm^n,\zm^n)\\
&\leq &  H(P_e) + P_e \log(2^{m-l}-1)
\end{eqnarray*}
since knowing whether the receiver could decode the worst noise does not affect the error
capability of $C_{ECC}$.
This concludes the proof.
\end{proof}

Let us compare the result of Lemmas \ref{lem:bound} and \ref{lem:boundactive}:
\begin{eqnarray*}
H(\xm^n|\am^l,\zm^n)&\geq & \min\{H(\um^{m-l}), H(\xm^n)+H(\vm^n)\}+ \min\{H(\vm^n),H(\xm^n) \}- \delta(C_{ECC}),\\
H(\xm^n|\am^l,\zm^n, F_d) &\geq & \min \{H(\um^{m-l}), H(\xm^n)+ H(\vm^n|F_d)\}+\min \{H(\vm^n|F_d), H(\xm^n) \}-\delta(C_{ECC}),
\end{eqnarray*}
where
\[
\delta(C_{ECC})  =  H(P_e)+P_e \log(2^{m-l}-1)\rightarrow 0,
\]
since $P_e\rightarrow 0$. As expected, the equivocation in the
case of an active adversary is smaller than for a passive
adversary, since $H(\vm^n|F_d)\leq H(\vm^n)$.

Based on the above, we easily get a counterpart of Theorem \ref{theor:pass} for the case of an active adversary.

\begin{theor}\label{theor:act}
When ${\rm Pr}(V_i^{(j)}=0) \neq {\rm Pr}(V_i^{(j)}=1) \neq 1/2$,
$i=1,2,...,n$, $j=1,2,...,\tau$, there exists a threshold $\tau_{thres,act}$
such that
\[
H(\km|\am^{\tau l},\zm^{\tau n},F_d) \;\;
\left\{ \begin{tabular}{lll}
$> 0$ & {\rm for} & $\tau < \tau_{thres,act}$ \\
$\rightarrow 0$   & {\rm for} & $\tau \geq \tau_{thres,act} \;.$\\
\end{tabular}\right.
\]
We have that $\tau_{thres,act}<\tau_{thres}$, the threshold for a passive adversary.
\end{theor}

\begin{proof}
When $\tau = 1$, $\xm^{(1)}=\xm^n =f^{(1)}(\km)$, Lemma \ref{lem:bound}
directly implies that $H(\xm^n|\am^l,\zm^n,F_d)=H(\km|\am^{l},\zm^{n},F_d) > 0$ is
achievable.

When $\tau > 1$ grows, we know from (\ref{eq:HKAZ}) that
\begin{eqnarray}
&&H(\km|\am^{\tau l},\zm^{\tau n},F_d)\nonumber \\
&=& H(\um^{\tau(m-l)}|\am^{\tau l},\zm^{\tau n},F_d) +
   H(\vm^{\tau n}|\am^{\tau l},\um^{\tau(m-l)},\zm^{\tau n},F_d)\nonumber \\
&& -H(\um^{\tau(m-l)}|\am^{\tau l},\km,\zm^{\tau n},F_d).
\end{eqnarray}
Now it is shown in the proof of Theorem \ref{theor:pass} that every term tends to zero,
using a decoding argument, which will hold similarly here, since the knowledge of $F_d$
cannot make the decoding more difficult.
\end{proof}

%
%

\section{Practical Implications and Applications Issues}
\label{sec:appli}

This section provides a generic discussion of the usefulness and
possible applications of the proposed approach.

\subsection{Implications of the security evaluation}

The analysis given in Sections \ref{sec:passive} and \ref{sec:active} shows
that in systems where the encoding-encryption paradigm is employed,
involvement of pure randomness via concatenation of dedicated wire-tap
and error-correction coding (instead of error-correction only)
provides an increased cryptographic security, by combining pseudo-randomness,
randomness and coding, which in a known-plaintext cryptanalytic scenario
implies an increased resistance against threats on the secret key.

The performed information-security evaluation more precisely points out
the following desirable security properties of the proposed approach:
(i) When the sample available for cryptanalysis is below a
certain size, the scheme provides uncertainty about the secret
key; (ii) Complexity of the secret key recovery appears as a
highly computationally complex problem even if the
available sample is such that the posterior uncertainty about the
secret key tends to zero. The main consequence of (i) is that even if
exhaustive search were to be employed for the secret key recovery,
a (large) number of candidates will appear. The statement
(ii) is an implication of the proofs of Theorems 1 and 2, where
the reduction to zero of the posterior uncertainty about the
secret key appears assuming employment of a decoding which has
complexity proportional to the exhaustive search over all possible
secret keys. Accordingly, the uncertainty tends to zero at the
expense of a decoding with exponential complexity. Actually, the
decrease of the uncertainty about the secret key with the increase
of the sample available for cryptanalysis appears as a consequence
of decoding capabilities of a low rate random binary block codes,
but at the expense of the decoding complexity which is exponential
in the secret key length.

The above features (i) and (ii) hold not only in a passive
attacking scenario where the attacker performs cryptanalysis based
on recording the ciphertext from a public communication channels,
but also in certain active attacking scenarios where the attacker
can modify the ciphertext and learn the effects of these
modifications.

\subsection{Framework for applications}

The encoding-encryption paradigm for secure and reliable
communications enjoys the following desirable properties: (i) When
the decryption is performed by bitwise XORing the keystream to the
ciphertext, an error in a bit before decryption causes an error in
the corresponding bit after decryption, without any
error-propagation, and (ii) Provides non-availability of the
error-free keystream when the communication channel is a noisy
one.

The proposed approach for enhancing the security of the
communications systems which follow encoding-encryption paradigm
could be employed in the design of these systems from scratch as
well as in upgrading of the existing ones.

In the case of upgrading the existing systems, the implementation
assumption is that the employed, already existing, binary linear
block error-correction code $(m,n)$ which encodes $m$ bits into a
codeword from $GF(2^n)$, could be replaced with a binary block
code $(m',n)$ with the same error correction capability but with
$m' > m$. Accordingly, $m'-m$ random bits can be concatenated with
$m$ information bits and mapped into the new $m$-bits via a
homophonic encoder. The obtained output from homophonic encoder is
the input for the error-correcting one. Taking into account the
notation from Section 2, the previous means that instead of
performing $C_{ECC}({\bf a})$ which is a linear mapping
$\{0,1\}^{m} \rightarrow \{0,1\}^{n}$, the following should be
performed: $C_{ECC}(C_H({\bf a} || {\bf u}))$ where ${\bf a} ||
{\bf u}$ is a concatenation of an $m$-dimensional vector and an
$m'-m$-dimensional one, $C_H(\cdot)$ is a linear mapping
$\{0,1\}^{m'} \rightarrow \{0,1\}^{m'}$ and $C_{ECC}(\cdot)$ is a
linear mapping $\{0,1\}^{m'} \rightarrow \{0,1\}^{n}$. On the
receiving side, the decoding procedures after decryption are
straightforward (see Fig. 2): The error correction decoding
removes the random errors, and the message ${\bf a}$ is obtained
by truncating of the inverse linear mapping corresponding to the
homophonic decoding.

In the case of a design of the encoding-encryption system from
the scratch, the design should include a coding box which performs
the concatenation of homophonic and error-correction coding in a
manner which fits the rate of the concatenated code to the given
constraints.

Note that from an implementation point of view, replacement of a
linear block encoding by a concatenation of a block linear
homophonic and error correction encoding is a replacement of one
binary matrix with another binary matrix which is the product of
the matrices corresponding to the homophonic and error-correction
encoders. Accordingly, the implementation complexity of two
concatenated codes could be approximately the same as the
implementation complexity of an error correcting code only.

%
%

\section{Conclusion}
\label{sec:conclusion}

The problem addressed in this paper is the one of enhancing
security of certain communications systems which employ error
correction encoding of the messages and encryption of the obtained
codewords in order to provide both secrecy and reliability of
the transmission. This paper yields a proposal for providing the
enhanced security of the considered systems employing randomness
and dedicated coding and the information-theoretic security
evaluation of the proposed approach. The analysis given in this
paper implies that in the systems where the encoding-encryption
paradigm is employed, the cryptographic security can
be enhanced via involvement of a homophonic coding based on pure
randomness as follows: Instead of just error-correction encoding
before the encryption, this paper proposes employment of a
concatenation of linear block homophonic and error-correction
encoding. The proposal and its cryptographic security evaluation
are given in a generic manner and accordingly yield a generic
framework for particular applications.

Note that the information-theoretic consideration of the
cryptographic security yields a basic evaluation of the related
cryptographic features. On the other hand, the
information-theoretic security evaluation also provides
specification of the settings when it is possible to perform the
secret key recovery, but it does not specify and only indicates
the expected complexity of this problem. We show that with the aid
of a dedicated wire-tap encoder, the amount of uncertainty that
the adversary has about the key given all the information he could
gather during different passive or active attacks he can mount, is
a decreasing function of the sample available for cryptanalysis.
This means that the wire-tap encoder can indeed provide an
information theoretical security level over a period of time,
after which a large enough sample is collected and the function
tends to zero, entering a regime in which a computational security
analysis is needed.

An interesting issue for further work is the characterization of
the transition region in which the uncertainty drops from a
certain value to close to zero. Also, because after all, the
uncertainty tends to zero, the computational complexity based
evaluation of cryptographic security is a direction for a future
work.

%
%

\section*{Acknowledgments}

The research of F. Oggier is supported in part by the Singapore National
Research Foundation under Research Grant NRF-RF2009-07 and
NRF-CRP2-2007-03, and in part by the Nanyang Technological University
under Research Grant M58110049 and M58110070.
This work was done partly while M. Mihaljevi\'c was visiting the division of
mathematical sciences, Nanyang Technological University, Singapore, and
partly while F. Oggier was visiting the Research Center for Information
Security, Tokyo.

%
%

\end{document}